\documentclass[submission]{dmtcs-episciences}
\usepackage[utf8]{inputenc}
\usepackage{amsmath, amscd, amssymb, amsthm, latexsym,dsfont}
\usepackage{hyperref, textcomp,url}
\numberwithin{equation}{section}
\theoremstyle{plain}
\newtheorem{theorem}{Theorem}[section]
\newtheorem{lemma}[theorem]{Lemma}
\newtheorem{proposition}[theorem]{Proposition}

\theoremstyle{definition}
\newtheorem{definition}[theorem]{Definition}
\newtheorem{notation}[theorem]{Notation}
\newtheorem{remark}[theorem]{Remark}

\newcommand\+[1]{\mathcal{#1}}

\newcommand\GS[2] { {\gamma}_{#1\rightarrow #2}}

\newcommand{\Z}{{\mathbb Z}}

\newcommand{\bz}{{\bf 0}}
\def\set#1{\{#1\}}
\def\pair#1{\langle{#1}\rangle}
\def\bz{{\bf 0}}
\def\S{\Sigma}

\def\restrict#1{\raise-.5ex\hbox{\ensuremath|}_{#1}}
\def\ct{{\mathcal T}}

 \usepackage[round]{natbib}
    
\author{Andr\'e Arnold\affiliationmark{1} 
\and Patrick C\'egielski\affiliationmark{2} 
\and Serge Grigorieff\affiliationmark{3}
\and Ir\`ene  Guessarian\affiliationmark{3,4}}
\title{The algebra of binary trees is affine complete}
\affiliation{Universit\'e  de Bordeaux, France.\\
LACL, 
Universit\'e  Paris XII -- IUT de S\'enart-Fontainebleau, France.\\ 
IRIF,  CNRS \& Universit\'e Paris-Diderot, France. \\
Emerita Sorbonne Universit\'e, Paris}
\received{2020-11-10}
\revised{2021-03-08}
\accepted{2021-05-11}
\begin{document}
\publicationdetails{23}{2021}{2}{1}{6890}
\maketitle

\begin{abstract}
A function on an algebra  is congruence preserving if, for any
congruence, it maps pairs of congruent elements onto pairs of
congruent elements. We show that on the algebra of  binary trees  
whose leaves are labeled by letters of an alphabet containing at
least three  letters, a function  is  congruence preserving if and only
if it is a polynomial function, thus exhibiting the first example of
a non commutative and non associative affine complete algebra.
\end{abstract}
\begin{center}
{\Large\it Merci \`a Maurice pour nombre de discussions alg\'ebriques passionnantes}
\end{center}
\section{Introduction}
A function on an algebra  is congruence preserving if, for any
congruence, it maps pairs of congruent elements onto pairs of
congruent elements.

A polynomial function on an algebra is a function defined by a term
of the algebra using variables, constants and the operations of the
algebra. Obviously, every polynomial function is  congruence
preserving.  

Algebras where all congruence preserving functions are polynomial
functions are called {\em affine complete} in the terminology
introduced by \cite{werner1971}. They are extensively studied in
the book by \cite{KaarliPixley}.

In the commutative case,  many  algebras have been shown to be
affine complete: Boolean algebras \citep{gratzer1962}, $p$-rings
with unit \citep{iskander1972}. For distributive lattices, 
\cite{plos} described congruence preserving functions, and
\cite{gratzer1964} determined which distributive lattices are affine
complete. Affine completenes is an intrinsic  property of an
algebra, which fails to hold even for very simple algebras: e.g., in
$\+A = \langle \Z, + \rangle$, the function
$f \colon \Z \to \Z$ defined by
\[
f(x) 
= \texttt{if $x\geq0$ then $\dfrac{\Gamma(1/2)}{2\times4^x\times x!}
\int_1^\infty e^{-t/2}(t^2-1)^x dt$
else $-f(-x)$.}
\]
has been proved to be congruence preserving \citep{cgg15}, but it
is {\em not a polynomial function} because its power series is
infinite. Hence $\+A = \langle \Z, + \rangle$ is not affine complete.

In the non commutative case, very little is known about affine
complete algebras. We proved in \cite{acgg20} that the free
monoid $\Sigma^*$ is an associative non commutative affine
complete algebra if $\Sigma$ has at least three letters, and we
proved in \cite{acgg20} a partial result concerning a non
commutative and non associative algebra: every {\em unary
congruence preserving} function $f \colon T(\Sigma) \to T(\Sigma)$
is a polynomial function, where $T(\Sigma)$ is the algebra of
{\em full} binary trees with leaves labelled by letters of an alphabet
$\Sigma$ having at least three letters. We here generalize this
result proving that a congruence preserving function 
$f \colon \+T(\Sigma)^n \to \+T(\Sigma)$ of any arity $n$ is a
polynomial function, where $\+T(\Sigma)$ is the algebra of arbitrary
(possibly non full) binary trees with labelled leaves. This
generalization is twofold: (1) non full binary trees are allowed in
$\+T(\Sigma)$, and (2) congruence preserving functions of arbitrary
arity are allowed. This exhibits an example of a non commutative
and non associative affine complete algebra. Non commutative and
non associative algebras are of constant use in Computer Science,
and congruences are also very often used, whence the potential
usefulness of our result.

We first define binary trees and their congruences, we then study
conditions which will enable us to prove that every congruence
preserving function is a polynomial  function, and to finally prove the
affine completeness of $T(\Sigma)$.
\section{The algebra of  binary trees}

\subsection{Trees, congruences}
For an algebra $\+A$  with domain $ A$, a {\em congruence}
$\sim$ on  $\+ A$ is an equivalence relation on $A$ which is
compatible with the operations of $\+A$. We state the
characterization of congruences by kernels of homomorphisms.
 
\begin{lemma}\label{l:ker} 
Let $\+A = \langle A\,,\, \star \rangle$ be an algebra with a binary
operation $\star$. An equivalence $\sim$ on $A$ is a congruence
iff there exists an algebra $\+B = \langle B\,,\, * \rangle$ with a
binary operation $*$ and there exists $\theta \colon A \to B$ a
homomorphism such that $\sim$ coincides with the kernel
congruence $\ker(\theta)$ of $\theta$, defined by
$x \sim_\theta y$ iff $\theta(x) = \theta(y)$.
\end{lemma}

Let $\Sigma$ be an alphabet not containing $\{0,1\}$. We shall
represent the algebra of binary trees over $\S$, i.e., trees with
leaves labeled by letters of $\Sigma$, as a set of words
$\+T(\Sigma)$ on the alphabet $\S \cup \{0,1\}$, together with the
binary product operation $\star$.
 
\begin{definition}\label{l:BTLL}
The algebra $\+ {B} = \langle \+T(\Sigma), \star \rangle$ of binary
trees over $\Sigma$ is defined as follows.
\begin{itemize}
\item A binary tree over $\Sigma$ is a finite set of words
$t\subseteq \set{0,1}^*\Sigma$ such that: For any $ua, vb \in t$, if
$ua \neq vb$ then $u$ is not a prefix of $v$ and $v$ is not a prefix
of $u$. The carrier set $\+T(\Sigma)$ is the set of all binary trees.
The empty set $\emptyset$ is a binary tree denoted by $\bz$.
\item The binary product operation $\star$ is defined by: for
$t,\; t' \in \+T(\Sigma)$, $t \star t' = 0.t \cup 1.t'$. In particular,
$\bz \star \bz = \bz$.
\end{itemize}
\end{definition}

\begin{figure}
\centering

\newcommand{\lb}{\begin{picture}(10,20)(0,0) \put(10,20){\circle*{2}}
     \put(0,0){\line(1,2){10}}\end{picture}}
 
\newcommand{\rb}{\begin{picture}(10,20)(0,0)
     \put(0,20){\circle*{2}}
     \put(10,0){\line(-1,2){10}}\end{picture}}
 
  \setlength{\unitlength}{1.6\unitlength}
  
\begin{picture}(250,50)
  \put(0,10){\lb}\put(10,30){\lb}\put(20,30){\rb}
  \put(0,5){$b$}\put(30,25){$a$}
  
  \put(50,20){\lb}\put(60,20){\rb}
  \put(50,15){$c$}\put(70,15){$d$}

  \put(90,10){\lb}\put(100,30){\lb}\put(110,10){\lb}\put(110,30){\rb}\put(120,10){\rb}
  \put(90,5){$b$}\put(110,5){$c$}\put(130,5){$d$}
  
  \put(150,10){\lb}\put(160,30){\lb}\put(160,10){\rb}\put(170,30){\rb}\put(180,10){\rb}
  \put(150,5){$b$}\put(170,5){$c$}\put(190,5){$d$}
  
  \put(210,10){\lb}\put(220,30){\lb}\put(230,10){\lb}\put(230,30){\rb}\put(240,10){\rb}
  \put(210,5){$a$}\put(230,5){$b$}\put(250,5){$c$}
\end{picture}

\caption{
From left to right, $t = \{00b,1a\}$, $\tau = \{0c,1d\}$,
$t_1 = {\GS{a}{\tau}}(t) = \{00b,01c,11d\}$, $t_2 = \{00b,01c,11d\}$,
$t_3 = \{00a,10b,11c\}$.
Trees $t_1, \ t_2, \ t_3$ have the same size 6, trees $t_1$ and
$t_3$  are similar (have the same skeleton.)
}
\label{fig:tree}
\end{figure}
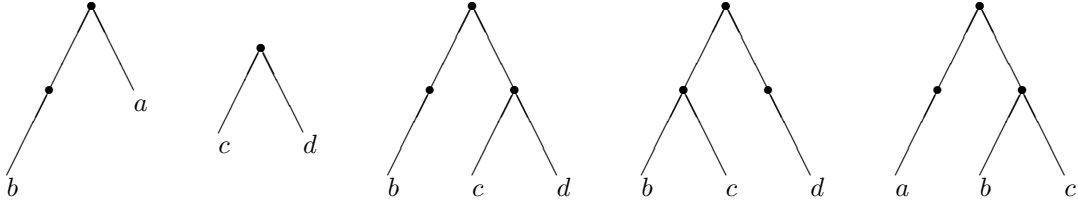

When the alphabet $\S$ is clear, we will denote by $\+T$ the set of
all binary trees. Trees are generated by $\set{\bz} \cup \Sigma$ and
the operation $\star$.

An essential property of this algebra $\+B$ is that its elements are
uniquely decomposable.

\begin{lemma}[Unicity of  decomposition]\label{uniquedecompo}
If $t$ is a tree not in $\set{\bz} \cup \Sigma$ then there exists a
unique ordered pair $\pair{t_1, t_2} \neq \pair{\bz, \bz}$ in $\+T^2$
such that $t = t_1 \star t_2$.
\end{lemma}

This property allows us to associate with each $t\in \+T$ its
{\em size} $|t|$ (number of nodes)

-- $|\bz| = 0 $, and for all $a \in \Sigma$, $|a| = 1$,

-- if $t \notin \set{\bz} \cup \Sigma$ then $t = t_1 \star t_2$, and
$|t| = |t_1| + |t_2| + 1$.

\medskip

If $|t| > 1$ then there exist $t_1, t_2 $ with $|t_i| < |t|$ such that
$t = t_1\star t_2$. Trees $t \star\ t'$, $\bz \star t'$, $t \star \bz$ are
trees whose root has two sons, a single right son, a single left son,
respectively. See Figure~\ref{fig:tree}. 

\subsection{Homomorphisms, graftings}

\begin{lemma}\label{l:ref1}
Let $\+B = \langle B\,,\, * \rangle$ be an algebra with a binary
operation $*$. Every mapping $h \colon \S \to B$ can be uniquely
extended to a homomorphism $h \colon \+T \to B$. 
\end{lemma}

\begin{remark} 1) Because of the universal property of Lemma
\ref{l:ref1}, homomorphisms are (uniquely) defined by giving their
values on $\Sigma$.

2) For every endomorphism, $h(\bz) = \bz$. Otherwise, as
$\bz = \bz \star \bz$, $h(\bz) = h(\bz) \star h(\bz)$; if $h(\bz) = t$ with
$|t| \geq 1$ then $t = t \star t$ implies $|t| = 2|t| + 1$, a
contradiction.
\end{remark}

\begin{definition}
For a given $a \in \Sigma$, let $\nu_a$ be the endomorphism
sending $\Sigma$ onto $a$. If for some $a \in \Sigma$,
$\nu_a(t) = \nu_a(t')$, trees $t$ and $t'$ are said to be {\em similar},
which is denoted by $t \sim_s t'$.
\end{definition}

Note that the congruence $\sim_s$ does not depend on the choice
of the letter $a \in \Sigma$ since $\nu_b(t) = \nu_b(\nu_a(t))$. From
an intuitive viewpoint, $t \sim_s t'$ means that $t$ and $t'$ have the
same skeleton, i.e., they are identical except for the leaf labels.
See Figure~\ref{fig:tree}.

Other congruences fundamental for our proof are the kernels of the
grafting endomorphisms, defined below. 

\begin{definition}[Grafting]\label{grafting} 
Let $a \in \Sigma$ and $\tau \in \+T$. Then the grafting
${\GS{a}{\tau}} \colon \+T \to \+T$ is the endomorphism defined by
its restriction on $\Sigma$
\[
{\GS{a}{\tau}}(b) = \begin{cases} \tau & \text{ if  }b = a,\\
                                                       b &\text{ if  }b \neq a.
\end{cases}
\]
\end{definition}

In other words, for any $a \in \Sigma$ and any $\tau \in \+T$,
$\GS{a}{\tau}$ is the endomorphism sending the letter $a$ on
$\tau$ and each other letter on itself.

An endomorphism $h$ of $\langle \+T(\Sigma), \star \rangle$ is
{\em idempotent} if for every $t \in \+T$, $h(h(t)) = h(t)$. By Lemma
\ref{l:ref1}, $h$ is idempotent iff for every $a \in \S$, $h(h(a)) = h(a)$.
For instance if $a$ does not occur in $\tau$ then $\GS{a}{\tau}$ is
idempotent.

\begin{proposition}\label{p:G2}
Let $\tau \in \+T$, let $t,\ t' \in \+T$, and let $a_1 \neq a_2$ be two
letters in $\Sigma$. If $ {\GS{a_i}{\tau}}(t) = {\GS{a_i}{\tau}}(t')$ for
$i = 1, 2$, then $t = t'$.
\end{proposition}

\begin{proof}
By induction on $\min(|t|, |t'|)$. 

{\em Basis Case 0:} If $\min(|t|, |t'|) = 0$ then one of $t, t'$ is $\bz$,
say $t = \bz$. If $t' \neq \bz$ then $t'$ contains at least one
occurrence of some letter $b$. As
$\GS{a_i}{\tau}(t') = \GS{a_i}{\tau}(t) = \GS{a_i}{\tau}(\bz) = \bz$, we
have $\GS{a_i}{\tau}(t') = \bz$, which implies (because $t' \neq \bz$
was supposed) that $\tau = \bz$. Then $\GS{a_i}{\tau}(t') = \bz$
implies that all leaves of $t'$ are equal to both $a_1$ and $a_2$,
a contradiction. Hence $t' = \bz$ and $t = t'$.

{\em Basis Case 1:} If $\min(|t|, |t'|) = 1$ then $t$ or $t'$ is a letter,
say $t = b$, and there is one $i$, say $i = 1$, such that
$a_1 \neq b$, thus $b = {\GS{a_1}{\tau}}(t) = {\GS{a_1}{\tau}}(t')$. 
  
\begin{itemize} 
\item If $t'$ is a letter $c\neq b$, then ${\GS{a_1}{\tau}}(c) = b$. If
$c = a_1$ then $b = {\GS{a_1}{\tau}}(c) = \tau$. Since
$\GS{a_2}{\tau}(c) = c = \GS{a_2}{\tau}(b) \in \set {\tau,b} = \set {b}$,
we have that $c = b$, a contradiction. If $c \neq a_1$ and
${\GS{a_1}{\tau}}(c) = c \neq b = {\GS{a_1}{\tau}}(c)$, a
contradiction. Hence $t' = t = b$.

\item If $|t'| > 1$ then $t' = t'_1 \star t'_2$, and
${\GS{a_1}{\tau}(t')} = {\GS{a_1}{\tau}(t'_1)} \star {\GS{a_1}{\tau}(t'_2)}$
which can be only of size 0 or $\geq 2$, contradicting
${\GS{a_1}{\tau}}(t') = b$. this case is excluded.
\end{itemize}

\smallskip

{\em Induction:} If $\min(|t|, |t'|) > 1$ then $t = t_1 \star t_2$ and
$t' = t'_1 \star t'_2$ with $\min(|t_i|, |t'_i|) < \min(|t|, |t'|)$, for
$i = 1, 2$. By Lemma \ref{uniquedecompo},
${\GS{a_j}{\tau}}(t_1) \star {\GS{a_j}{\tau}}(t_2)
 = {\GS{a_j}{\tau}}(t'_1) \star {\GS{a_j}{\tau}}(t'_2)$ implies
${\GS{a_j}{\tau}}(t_i) = {\GS{a_j}{\tau}}(t'_i)$, for $j = 1, 2$. By the
induction hypothesis $t_i = t'_i$, hence $t = t'$.
\end{proof}
 
\begin{proposition}\label{p:tt'sim}
Let us fix $a \in \Sigma$, with $|\Sigma| \geq 3$, $t,\ t' \in \+T$
such that $t\sim_s t'$.

(1) If, for some $\tau \in \+T$ of size $|\tau| \neq 1$,
${\GS{a}{\tau}}(t) = {\GS{a}{\tau}}(t')$, then $t = t'$.

(2) If, for all $b \neq a$, $b \in \Sigma$,
${\GS{a}{b}}(t) = {\GS{a}{b}}(t')$, then $t = t'$.
\end{proposition}

\begin{proof}
Both (1) and (2) are proved by induction on $|t| = |t'|$, and in both
cases, the result obviously holds if $t = t' = \bz$.
 
\medskip

\noindent {\em Basis:} If $|t| = |t'| =1$.
   
(1) We assume that $t = b \neq c = t'$.
   
\noindent (i) If $a \not\in \{b,c\}$ then
${\GS{a}{\tau}}(t) = b \neq c = {\GS{a}{\tau}}(t')$, a contradiction.
 
\noindent (ii) Otherwise, $a \in \{b,c\}$, e.g., $a = b = t$, then 
${\GS{a}{\tau}}(t) = {\GS{a}{\tau}}(a) = \tau$ and
${\GS{a}{\tau}}(t') = {\GS{a}{\tau}}(c) = c$, hence $\tau = c$, which
contradicts $|\tau| \neq 1$.
  
(2) We assume that $t = b \neq c = t'$.
  
\noindent (i) The case $a \not\in \{b, c\}$ yields a contradiction as
in case (1).
  
\noindent (ii) Otherwise, e.g., $a = b$, there exists $d \not\in \{a, c\}$,
and we get $ {\GS{a}{d}}(t) = {\GS{a}{d}}(a) = d$ and
${\GS{a}{d}}(t') = {\GS{a}{d}}(c) = c$, contradicting
${\GS{a}{d}}(t) = {\GS{a}{d}}(t')$.
  
\medskip

\noindent {\em Induction:} As in Proposition \ref{p:G2} in both cases:
since $t$ and $t'$ are similar, $t = t_1 \star t_2 $ and
$t' = t'_1 \star t'_2$ with $t_i$ similar to $t'_i$ and $|t_i| < |t'_i|$.
\end{proof}
 
\subsection{Congruence preserving functions on  trees}
\label{sec:congr-pres-funct}

\begin{definition}\label{def1:cp}
A function $f \colon \+T^n \to \+T$ is {\em congruence preserving}
(abbreviated into CP) if for all congruences $\sim$ on~$\+T$, for
all $t_1, \ldots, t_n,\ t'_1, \ldots, t'_n$ in $\+T$, $t_i \sim t'_i$ for all
$i = 1, \ldots, n$, implies $f(t_1, \ldots, t_n) \sim f(t'_1, \ldots, t'_n)$.
\end{definition}

\begin{remark}\label{r:hfh=fh}
(1) It follows from Lemma \ref{l:ker} that CP functions are
characterized by the fact that for all homomorphisms $h$ from
$\pair{\ct, \star}$ to any algebra $\pair{A, \star_A}$, $h(t_i) = h(t'_i)$
for all $i =1, \ldots, n$, implies
$h(f(t_1, \ldots, t_n)) = h( f(t'_1, \ldots, t'_n))$.
 
(2) If $f$ is CP and endomorphism $h$ is idempotent then
$h(f(t_1, \ldots, t_n)) = h(f(h(t_1), \ldots, h(t_n)))$. Indeed, let
$\sim_h$ be the congruence associated with $h$, for
$i = 1, \ldots, n$, we have $h(t_i) = h(h(t_i))$, hence
$t_i \sim_h h(t_i)$, whence the result.
\end{remark}

\noindent We will show that congruence preserving functions on the
algebra $\langle \+T(\Sigma), \star \rangle$ are polynomial functions.
Let us first formally define polynomials on trees.
 
\begin{definition}
Let $x_1, \ldots, x_n \not \in \Sigma$ be called {\em variables}. A
{\em polynomial} $P(x_1, \ldots, x_n)$ is a tree on the alphabet
$\Sigma \cup \{x_1, \ldots, x_n\}$.

With every polynomial $P(x_1, \ldots, x_n)$ we will associate a
{\em polynomial function} $\tilde P \colon \+T^n \to \+T$ defined
by: for any
$\vec u = \langle t_1, \ldots, t_i, \ldots, t_n \rangle \in \+T^n$,

$\tilde P(\vec u) = \left\{
    \begin{array}{ll}
      P  & \mbox{if $P = \bz$ or $P \in \Sigma$}\\
      t_i& \mbox{if $P = x_i$}\\
      \widetilde {P_1}(\vec u) \star \widetilde {P_2}(\vec u)& \mbox{if $P = P_1 \star P_2$}
    \end{array}
\right.$
\end{definition}

Obviously, every polynomial function is CP. Our goal is to prove the
converse, namely

\begin{theorem}\label{E}
Let $|\Sigma| \geq 3$. If $g \colon \+T^n \to \+T$ is CP then there
exists a polynomial $P_g$ such that $g = \widetilde {P_g}$.
\end{theorem}

\section{Equality of CP functions}
\label{sec:equal-cp-funct}
\begin{notation}
For any $f \colon \ct^n \to \ct$, we denote by $f\restrict {\Sigma^n}$
its restriction to $\Sigma^n$.
\end{notation}

In this section we prove that if $f$ and $g$ are two CP functions,
then $f\restrict{\Sigma^n} = g\restrict{\Sigma^n}$ implies $f = g$,
{\em provided that $\Sigma$ contains at least three letters.}

\begin{lemma}\label{similar} 
Suppose $\Sigma$ has at least three letters. If $f$ and $g$ are
unary CP functions on $ \+T$ such that  for all $a \in \Sigma$,
$f(a) = g(a)$ then for all $t \in \+T$, $f(t)$ and $g(t)$ are similar.
\end{lemma}

\begin{proof} We have to show that $\nu_a(f(t)) = \nu_a(g(t))$ for
some $a \in \Sigma$ and for all $t$. As $\nu_a$ is idempotent
and $f$ is CP, by Remark \ref{r:hfh=fh}~(2),
$\nu_a(f(t)) = \nu_a(f(\nu_a(t)))$, and similarly for $g$. Hence it
suffices to prove $f(\nu_a(t)) = g(\nu_a(t))$. Let
$b_1, \ b_2 \in \Sigma$ such that $a,\ b_1,\ b_2$ are pairwise
distinct. As $\GS{b_i}{\nu_a(t)}$ is idempotent, by Remark
\ref{r:hfh=fh}~(2), we have
$\GS{b_i}{\nu_a(t)}(f(b_i)) = \GS{b_i}{\nu_a(t)}(f(\nu_a(t)))$. The
same holds for $g$, i.e.,
$\GS{b_i}{\nu_a(t)}(g(b_i)) = \GS{b_i}{\nu_a(t)}(g(\nu_a(t)))$. From
$f(b_i) = g(b_i)$, we deduce that
$\GS{b_i}{\nu_a(t)}(f(\nu_a(t))) = \GS{b_i}{\nu_a(t)}(g(\nu_a(t)))$.
This equality holds for $i = 1, 2$, thus Proposition \ref{p:G2}
implies that $f(\nu_a(t)) = g(\nu_a(t))$.
\end{proof}

\noindent The following proposition shows that a unary CP function
$f$ is completely determined by its values on $\Sigma$.

\begin{proposition} \label{unary}
Suppose $\Sigma$ has at least three letters. If $f$ and $g$ are
unary CP functions on $\+T$ such that for all $a \in \Sigma$,
$f(a) = g(a)$ then for all $t \in \+T$, $f(t) = g(t)$.
\end{proposition}

\begin{proof} 
Let $a$ be a letter that occurs in $t$. For any other letter $b$, the
endomorphisms $\GS{a}{b}$ and $\GS{a}{t_b}$ are idempotent,
where $t_b = \GS{a}{b}(t)$. Thus by Remark \ref{r:hfh=fh}~(2),
$\GS{a}{t_b}(f(a)) = \GS{a}{t_b}(f(t_b))$, and
$\GS{a}{t_b}(g(a)) = \GS{a}{t_b}(g(t_b))$. As $f(a) = g(a)$  we have
$\GS{a}{t_b}(f(t_b)) = \GS{a}{t_b}(g(t_b))$. By Lemma \ref{similar},
$f(t_b)$ and $g(t_b)$ are similar, and by Proposition \ref{p:tt'sim}
(1) $f(t_b) = g(t_b)$.

On the other hand, as $f$ and $g$ are CP and
$t \sim_{\GS{a}{b}} t_b$, we get $\GS{a}{b}(f(t)) = \GS{a}{b}(f(t_b))$
and $\GS{a}{b}(g(t)) = \GS{a}{b}(g(t_b))$, hence
$\GS{a}{b}(f(t)) = \GS{a}{b}(g(t)) $. As this is true for all $b \neq a$,
we have by Proposition \ref{p:tt'sim} (2), $f(t) = g(t)$.
\end{proof}

Proposition \ref{unary} now can be generalized.

\begin{notation}
For any function $f \colon \+T^{n+1} \to \+T$, any $t \in \+T$, and
$\vec u = \langle t_1, \ldots, t_n \rangle$, we define 

\noindent (1) a $n$-ary function $f_{\cdots,t}$ obtained by
``freezing'' the (n+1)th argument to the value $t$, and defined by:
for all $\vec u \in \+T^{n}$, $f_{\cdots,t}(\vec u) = f(\vec u,t)$,

\noindent (2) a unary function $f_{\vec u,\cdot}$ obtained by
``freezing'' the $n$ first arguments to the value
$\vec u =\langle t_1, \ldots, t_n \rangle$, and defined by: for all
$t\in\+T$, $f_{\vec u,\cdot}(t) = f(\vec u,t)$.
\end{notation}

\begin{proposition}\label{andre}
Let $f$ and $g$ be n-ary CP functions on $\+T$ such that for all
$a_1, \ldots, a_n \in \Sigma$,
$f(a_1, \ldots, a_n) = g(a_1, \ldots, a_n)$ then for all
$t_1, \ldots, t_n \in \+T$, $f(t_1, \ldots, t_n) = g(t_1, \ldots, t_n)$.
\end{proposition}

\begin{proof}
By induction on $n$. For $n = 1$ the result was proved in
Proposition \ref{unary}. Assume  the result holds for
$n$. By the hypothesis, for all $a_1, \ldots, a_n , a \in \Sigma$, we
have $f(a_1, \ldots, a_n, a) = g(a_1, \ldots, a_n, a)$, i.e.,
$f_{\cdots,a}(a_1, \ldots, a_n) = g_{\cdots,a}(a_1, \ldots, a_n)$.
By the induction applied to $f_{\cdots,a}$, for all
$\vec u \in \+T^n$, $f_{\cdots, a}(\vec u) = g_{\cdots, a}(\vec u)$, or
equivalently $f_{\vec u, \cdot}(a) = g_{\vec u, \cdot}(a)$. As
$f_{\vec u, \cdot}(a) = g_{\vec u, \cdot}(a)$, applying now
Proposition \ref{unary} to $f_{\vec u, \cdot}$ and $g_{\vec u, \cdot}$
yields $f_{\vec u, \cdot}(t) = g_{\vec u, \cdot}(t)$ for all $t$, hence
$f(\vec u, t) = g(\vec u, t)$.
\end{proof}

\section{The algebra of binary trees is affine complete}
To prove that any CP function is a polynomial  function, as a
consequence of Proposition~\ref{andre} and of the fact that a
polynomial function is CP, it is enough to show that the restriction
$f\restrict{\Sigma^n}$ of $f \colon \ct^n \to \ct$ to $\Sigma^n$ is
equal to the restriction $\tilde{P}\restrict{\Sigma^n}$ of a $n$-ary
polynomial function. For such restricted functions we introduce a
weakened version WCP of the CP condition, namely,

\begin{definition}\label{d:WCP}
Function $g \colon \+T^n \to \+T$ is said to be WCP iff for any
idempotent mapping $h \colon \Sigma \to \Sigma$,
$\forall \vec u, \vec v \in \Sigma^n$,
$h(\vec u) = h(\vec v) \implies h(g(\vec u)) = h(g(\vec v))$, where for
$\vec u = \langle u_1, \ldots, u_n \rangle$, $h(\vec u)$ denotes
$\langle h(u_1), \ldots, h(u_n) \rangle$. 
\end{definition}
 
Every CP function is clearly WCP.
 
\begin{lemma}
If $g$ is WCP then for all $\vec u, \vec v \in \Sigma^n$, $g(\vec u)$
and $g(\vec v)$ are similar. 
\end{lemma}
 
\begin{proof}
As $\nu_a(\vec u) = \nu_a(\vec v) = \langle a,\ldots, a \rangle$ for all
$\vec u, \vec v \in \Sigma^n$ and $g$ is WCP,
$\nu_a(g(\vec u)) = \nu_a(g(\vec v))$.
\end{proof}

We often use a different form of the condition WCP, which deals
only with alphabetic graftings.
 
\begin{proposition}
A function $g$ is WCP if and only if 
     
(GCP) (G for graftings) for all $a_1, a_2,\ldots, a_n \in \Sigma$,
$i \in \{1, \ldots, n\}$ and $b_i \in \Sigma$,  
$\GS{a_i}{b_i}(g(a_1, \ldots, a_n))
 = \GS{a_i}{b_i}(g(a_1, \ldots, a_{i-1}, b_i, a_{i+1}, \ldots, a_n))$.
\end{proposition}
 
\begin{proof}
Since
$\GS{a_i}{b_i}( a_1, \ldots, a_n)
 = \GS{a_i}{b_i}( a_1, \ldots, a_{i-1}, b_i, a_{i+1}, \ldots, a_n )$,
clearly WCP implies GCP. The proof of the converse is by
induction on $n$. It is obviously true for $n = 0$.

Otherwise, let $h$ be a mapping $h \colon \Sigma \to \Sigma$ and
let $\vec u, \vec v \in \S^n$ such that $h(\vec u ) = h(\vec v)$, and let
$a, b \in \S$ such that $h(a) = h(b)$. By (GCP), we have
$\GS{a}{b}(g(\vec u, a)) = \GS{a}{b}(g(\vec u, b))$, hence
$h(\GS{a}{b}(g(\vec u, a))) = h(\GS{a}{b}(g(\vec u, b)))$.

\noindent But $h( \GS{a}{b}(c)) = \left\{
  \begin{array}{ll}
  h(c) &\mbox{if $c \neq a$}\\
  h(b) = h(a)& \mbox{if $c = a$}
  \end{array}\right.$
hence $h \circ \GS{a}{b} = h$. Therefore
$h(g(\vec u, a)) = h(g(\vec u, b))$, and by the induction applied
to $g_{\ldots, b}$,
$h(g(\vec u, a)) = h(g(\vec u, b)) = h(g(\vec v, b))$.
\end{proof}

Let us first study unary WCP functions whose restriction to
$\Sigma$ takes its values in $\Sigma$. 

\begin{proposition}\label{p2Andre}
Assume $|\Sigma| \geq 3$. Let $f \colon \+T \to  \+T$ be WCP and
such that $f(\Sigma) \subseteq \Sigma$. Then $f\restrict{\Sigma}$
is (1) either a constant function (2) or the identity.
\end{proposition}

\begin{proof}
If $f$ is not the identity there exist $a,\ b$, with $a \neq b$ and
$f(a) = b$. As $\GS{a}{b}(f(b)) = \GS{a}{b}(f(a)) = \GS{a}{b}(b) = b$,
we get $f(b) \in \{a, b\}$. 

For $c \not\in \{a, b\}$, $\GS{a}{c}(f(c)) = \GS{a}{c}(f(a)) = b$ implies
$f(c) = b$. It remains to prove that $f(b) = b$. From
$\GS{b}{c}(f(b)) = \GS{b}{c}(f(c)) = c$, we deduce that
$f(b) \in \{c, b\}$, hence $f(b) \in \{a, b\} \cap \{c, b\} = \set{b}$, which
concludes the proof.
\end{proof}

We now will generalize Proposition \ref{p2Andre} by Proposition
\ref{pr3Andre} (replacing a unary $f$ by a $n$-ary $g$). 

\begin{proposition}\label{pr3Andre}
Assume $|\Sigma| \geq 3$. If $g \colon \+T^n \to  \+T$ is WCP and
such that $g(\S^n) \subseteq \Sigma$, then $g\restrict{\Sigma^n}$
is (1) either a constant function (2) or a projection $\pi^n_i$.
\end{proposition} 
 
\begin{proof}
The proof is by induction on $n$. By Proposition \ref{p2Andre} it is
true for $n = 1$. If $g$ is of arity $n+1$ then, by induction
hypothesis, for any $a \in \Sigma$, the function $g_{\cdots, a}$ of
arity $n$ is either a constant or a projection $\pi_i^n$. We first show
that these functions are all equal to a given $\pi_i^n$, or all equal to
a same constant, or every $g_{\cdots, a}$ is the constant function
$a$.

Let us assume that $g_{\cdots,a} = \pi_i^n$. Let
$\vec u = \pair{a, \ldots, a, c, a,\ldots, a}$ and
$\vec v = \pair{a, \ldots, a, d, a, \ldots, a}$ with $a ,c, d$ pairwise
disjoint, so that for any $b$, $\GS{a}{b}(g(\vec u, a)) = c$ and
$\GS{a}{b}(g(\vec v, a)) = d$. It follows from the GCP condition that
$\GS{a}{b}(g(\vec u, a)) = \GS{a}{b}(g(\vec u, b)) = c$ and
$\GS{a}{b}(g(\vec v, a)) = \GS{a}{b}(g(\vec v, b)) = d$,
which is impossible if $g_{\cdots, b}$ is either a constant or a
projection $\pi_j^n$ with $j \neq i$. Hence all $g_{\cdots, a}$ are
equal to $\pi_i^n$, implying $g = \pi_i^{n+1}$.
 
\smallskip
  
Assume now all the $g_{\cdots, a}$ are constant. For every
$\vec u, \vec v, a$, we have $g(\vec u, a) = g(\vec v, a)$. We
choose an arbitrary $\vec u \in \Sigma^n$ which will be fixed. By
the induction hypothesis $g_{\vec u, \cdot}$ is either (1) the
identity, or (2) a constant $c$. In case (1), for all $ \vec v, a$,
$g(\vec u, a) = g(\vec v, a) = a$ and $g = \pi_{n+1}^{n+1}$.
In case (2), for all $\vec v, a, b$,
$g(\vec u, a) = g(\vec v, b) = c$ and $g$ is a constant.
\end{proof}

As CP functions are WCP, for $g$ a CP function such that for some
$a_1, \ldots, a_n \in \Sigma $, $g(a_1, \ldots, a_n) \in \Sigma$, we
have shown that there exists a polynomial $P_g$, which is either a
constant $a \in \Sigma$ or an $x_i$, such that
$g = \widetilde {P_g}$. We will now extend to the case when
$g(a_1, \ldots, a_n) \not \in \Sigma$.

\begin{proposition}\label{prop:reduc}
Assume that $|\Sigma| \geq 3$. Let $g \colon \+T^n \to  \+T$ be
WCP. Then there exists a polynomial $P_g$ such that
$g\restrict{\Sigma^n} = \widetilde{P_g}\restrict{\Sigma^n}$.
\end{proposition}

\begin{proof}
Let $\sigma(g)$ be the common size of all the
$g(\vec u), \ \vec u \in \S^n$. The proof is by induction on
$\sigma(g)$.

{\em Basis:} If $\sigma(g) = 0 $ then
$g\restrict{\Sigma^n} = \tilde P\restrict{\Sigma^n} = \bz$. If
$\sigma(g) = 1$ then $g(a_1, \ldots, a_n) \in \Sigma$ and the
result is proved in Proposition \ref{pr3Andre}.

{\em Induction:} If $\sigma(g) > 1$ there exists two functions
$g_i \colon \ct^n \to \ct$ for $i = 1, 2$ such that for all
$\vec u \in \Sigma^n $, $g(\vec u) = g_1(\vec u) \star g_2(\vec u)$,
with $|\sigma(g_i)| < |\sigma(g)|$. It remains to show that both
$g_1$ and $g_2$ are WCP. Let $\vec u, \vec v \in \S^n$ be such
that $h(\vec u) = h(\vec v)$ for some mapping $h \colon \S \to \S$.
Extend $h$ as an endomorphism $\+T \to \+T$ by Lemma
\ref{l:ref1}, then
$h(g(\vec u)) = h(g_1(\vec u) \star g_2(\vec u))
 = h(g_1(\vec u)) \star h(g_2(\vec u))$.
Similarly, $h(g(\vec v)) = h(g_1(\vec v)) \star h(g_2(\vec v))$.
As $g$ is WCP and $h(\vec u) = h(\vec v)$, we have
$h(g(\vec u)) = h(g(\vec v))$. Then by Lemma \ref{uniquedecompo}
(unique decomposition) we get $h(g_i(\vec u)) = h(g_i(\vec v))$ for
$i =1, 2$. This is true for any $h$, thus $g_1$ and $g_2$ are WCP.
By the induction hypothesis there exists $P_i $ such
$\widetilde{P_i}\restrict{\Sigma^n} = {g_i}\restrict{\Sigma^n}$, hence
${g}\restrict{\Sigma^n}
 = \widetilde{P_1}\restrict{\Sigma^n}\star \widetilde{P_2}\restrict{\Sigma^n}
 = \widetilde{P_1 \star P_2}\restrict{\Sigma^n}$.
\end{proof}

\begin{theorem}
If $f \colon \ct^n\to \ct$ is CP then there exists a polynomial $P$
such that $f = \tilde{P}$.
\end{theorem}

\begin{proof}
Since $f$ is CP, $f$ also is WCP. By the previous proposition, there
exists $P$ such that
$f\restrict{\Sigma^n} = \tilde{P}\restrict{\Sigma^n}$, and by
Proposition~\ref{andre}, $f = \tilde{P}$.
\end{proof}

\section{Conclusion}

We proved that, when $\Sigma$ has at least three letters, the
algebra  of arbitrary binary trees with leaves labeled by letters of
$\Sigma$ is an affine complete algebra (non commutative and non
associative).

\acknowledgments We thanks the referees for comments which
helped to improve the paper.

\nocite{*}
\bibliographystyle{abbrvnat}

\end{document}